\documentclass[
a4paper,
twocolumn,
floatfix,
unpublished,
10pt]{quantumarticle}

\pdfoutput=1
\usepackage[numbers]{natbib}
\usepackage{bbm}
\usepackage{amsmath}
\usepackage{amssymb}
\usepackage{amsthm}
\usepackage{amsfonts}
\usepackage{amsmath}
\usepackage{graphicx}
\usepackage{xcolor}

\usepackage{mathtools}
\usepackage{booktabs} 
\usepackage{bm}
\usepackage[export]{adjustbox}
\usepackage{multirow}
\usepackage{mleftright}
\usepackage[normalem]{ulem}
\usepackage[ruled,lined]{algorithm2e}
\usepackage{tabularx}
\usepackage{physics}
\usepackage{orcidlink}
\usepackage{mathtools}
\usepackage{braket}

\newcommand{\beq}{\begin{equation}}
\newcommand{\eeq}{\end{equation}}
\newcommand{\bqa}{\begin{eqnarray}}
\newcommand{\eqa}{\end{eqnarray}}

\definecolor{maroon}{rgb}{0.7,0,0}

\definecolor{ngreen}{rgb}{0.3,0.7,0.3}

\definecolor{golden}{rgb}{0.8,0.6,0.1}

\definecolor{npurple}{rgb}{0.3,0,0.6}

\newcommand{\be}{\begin{equation}}
\newcommand{\ee}{\end{equation}}
\newcommand{\ba}{\begin{eqnarray}}
\newcommand{\ea}{\end{eqnarray}}

\newtheorem{observation}{Observation}
\newtheorem{definition}{Definition}
\newtheorem{thm}{Theorem}

\newtheorem{lem}{Lemma}

\newtheorem{prop}{Proposition}
\newtheorem{rem}{Remark}

\begin{document}

\title{Global Coherence in Quantum Discord as a Resource}  

\author{Chellasamy Jebarathinam}
   \affiliation{Physics Division, National Center for Theoretical Sciences, National Taiwan University, Taipei 106319, Taiwan, Republic of China}

\author{Huan-Yu Ku}
\email{huan.yu@ntnu.edu.tw}
\affiliation{Department of Physics, National Taiwan Normal University, Taipei 116059, Taiwan, Republic of China}

\author{Hao-Chung Cheng}
\email{haochung@ntu.edu.tw}
 \affiliation{Department of Electrical Engineering and Graduate Institute of Communication Engineering, National Taiwan University, Taipei 106319, Taiwan, Republic of China}

 \affiliation{Department of Mathematics, National Taiwan University, Taipei 106319, Taiwan, Republic of China}

 \affiliation{Center for Quantum Science and Engineering, National Taiwan University, Taipei 106319, Taiwan, Republic of China}
 \affiliation{Physics Division, National Center for Theoretical Sciences, National Taiwan University, Taipei 106319, Taiwan, Republic of China}
 \affiliation{
Hon Hai (Foxconn) Quantum Computing Center, New Taipei City 236, Taiwan, Republic of China}

\author{Hsi-Sheng Goan}

 \email{goan@phys.ntu.edu.tw}
 \affiliation{Department of Physics and Center for Theoretical Physics,
National Taiwan University, Taipei 106319, Taiwan, Republic of China}
\affiliation{Center for Quantum Science and Engineering, National Taiwan University, Taipei 106319, Taiwan, Republic of China}
 \affiliation{Physics Division, National Center for Theoretical Sciences, National Taiwan University, Taipei 106319, Taiwan, Republic of China}
   
\begin{abstract}
This work addresses which aspect of \textit{bipartite coherence} in \textit{quantum discord} is essential for \textit{genuinely quantum correlation}.
To this end,  \textit{global coherence} of bipartite states is defined as a form of bipartite coherence that is not local coherence in either of the subsystems or in both subsystems alone, and we identify it as being indicated by a witness of discord. With \textit{global coherence}  as \textit{a resource}, \textit{any local operations} of the form $\Phi_A \otimes \Phi_B$, which may create coherence locally, are \textit{free operations}. This implies that with global coherence as a resource for operational tasks, any local operations can be freely used, but require \textit{classical randomness not to be freely} available. Using this identification, we demonstrate that \textit{global coherence} is a \textit{necessary resource} for the task of \textit{semi-device-independent steerability} of quantum discord. On the other hand, for the task of \textit{remote state preparation} using \textit{quantum discord} in two-qubit states,   a \textit{necessary and sufficient quantum resource} is identified by invoking a \textit{witness of global coherence}.
\end{abstract}

\maketitle 
   
\section{Introduction}
Quantum coherence stems from the superposition principle of quantum theory and lies at the core of the quantumness of the states of a single quantum system. 
Recently, many studies have emerged exploring quantum coherence as a resource
\cite{BCP14,WY16,SAP17} and its applications for quantum information processing tasks \cite{Liu_2023,LLL+25}.
 Extending to many-body systems, quantum entanglement is an essential coherence feature. 
 In particular, one can always characterize and even quantify quantum coherence via entanglement by converting the coherence of a single quantum system into entanglement~\cite{SSD+15,ZMC+17,Chi17}.

 Quantum discord characterizes the quantumness of correlations beyond entanglement~\cite{OZ01,HV01}. 
 This initiated the exploration of the role of quantum discord as a resource for quantum information processing tasks using separable states~\cite{DSC08,LBA+08,DLM+12,MBC+12,ABC16}. 
 This is important from the perspective of more general quantum resources than entanglement and others~\cite{CG19,KTA+24} as well as experimental implementations in the presence of noise~\cite{Cd11}. 
 
 \textit{Genuinely quantum correlation} cannot be generated via local operations~\cite{VPR+97,VP98}. Genuinely quantum correlation is also present in discord without entanglement; however, certain discord without entanglement can be generated via local operations~\cite{GBH+10,BPP15}. 
 In~\cite{MYG+16}, it has been shown that in converting the coherence of a subsystem to quantum correlation under bipartite incoherent operations, discord without entanglement can also be generated. In this context, discord without entanglement has a \textit{genuinely quantum correlation} since the coherence of the subsystem is consumed in the process of bipartite incoherent operations, as also quantitatively studied in~\cite{MYG+16}.

 The present work is motivated by two key questions. 
\begin{enumerate}
    \item From the perspective of the foundations of discord \cite{Vedral2017}, it is relevant to ask whether there are operational tasks that require the aforementioned \textit{genuinely quantum correlation} present in discord without entanglement.
    \item  On the other hand,  from the perspective of using quantum discord as a resource \cite{Gio13}, the question arises as to whether the \textit{genuinely quantum correlation} can be defined within the context of a quantum resource theory.
\end{enumerate}

 To this end, we consider the \textit{witness of the discord} defined by the rank of the correlation matrix of the density matrix \cite{DVB10}. The \textit{bipartite coherence} of discordant states, \textit{whose discord is not indicated by this witness}, is due to \textit{local coherence} in a subsystem or in both subsystems. With this observation,  the notion of \textit{global coherence}  of bipartite states is operationally defined for any discordant state whose discord is indicated by the aforementioned witness of discord. We then introduce the resource theory of global coherence and demonstrate that \textit{any local operations} of the form $\Phi_A \otimes \Phi_B$, which may create coherence locally, are \textit{free operations}. To rigorously define the global coherence of bipartite states, in the context of the conversion of coherence into genuinely quantum correlation~ \cite{MYG+16}, we point out a strong link between the global coherence resource frameworks in the bipartite and single-partite states. 
 Then it is found that global coherence is necessary for \textit{semi-device-independent (SDI) steerability} \cite{JDS_PRA23}, just as entanglement is necessary for standard steerability formalized in~\cite{WJD07}. Finally, 
 to show how global coherence plays a role in the quantum communication task of \textit{remote state preparation (RSP)} using two-qubit states, we propose \textit{Mermin strength} as a witness for global coherence. We then demonstrate that any two-qubit state with a nonvanishing Mermin strength can be used to provide a quantum advantage for remote state preparation.

This article is organized as follows. In Sec.~\ref{Prlm}, we review quantum discord, and introduce semi-device-independent steering as preliminaries of this work. In Sec.~\ref{GCdefRT}, we formally define global coherence of bipartite states and introduce the resource theory of global coherence, and demonstrate its relationship with that of coherence of single-partite states. In Sec.~\ref{OTGR}, we then show the main result that such global coherence is necessary to demonstrate semi-device-independent steerability, and finally, as another main result, we identify the precise resource for the quantum communication task of remote state preparation using correlations in two-qubit states. In Sec.~\ref{concDisc}, we present conclusions and discussions.

\section{Preliminaries} \label{Prlm}
\subsection{Quantum discord}\label{ctoe}
Quantum discord captures a different form of quantumness in bipartite states that goes beyond quantum correlation due to entanglement \cite{OZ01,HV01}. 
In the following, we present the definition of quantum discord from Alice to Bob $D^{\rightarrow}(\rho_{AB})$. Before that, we introduce the definition of the classical correlation $C^{\rightarrow}(\rho_{AB})$ used in defining quantum discord. It is defined with respect to a measurement described by a positive operator-valued measure (POVM), \(\{M^A_i\}\), performed on subsystem \(A\). In such an asymmetric measurement scenario, we denote the probability of obtaining the outcome \(i\) on the measured subsystem as \(p_i = \tr\{ \rho_{AB} M^A_i   \otimes \mathbb{I}^B\} \)  and the corresponding post-measurement state for the subsystem \(B\) as \(\rho^B_{i} = \frac{1}{p_i} \tr_{A}\{\rho_{AB} M^A_i \otimes  \mathbb{I}_B \}\). Classical correlation $C^{\rightarrow}(\rho_{AB})$ for local measurements in subsystem $A$ is defined as the optimal amount of correlation in the bipartite state as 
 \be
 C^{\rightarrow}(\rho_{AB}) := S(\rho_B) - \min_{\{M^A_i\}} \left( \sum_i p_iS(\rho^B_{i})\right),
 \ee
where $S(\rho)$ is the Von Neumann entropy of a quantum state $\rho$, defined as $S(\rho)=-\tr{\rho\log_2\rho}$.
Quantum discord is then defined as 
\begin{align}\label{QDdef}
D^{\rightarrow}(\rho_{AB}):= I(\rho_{AB}) - C^{\rightarrow}(\rho_{AB}).
\end{align}            
Here $I(\rho_{AB})= S(\rho_A)+ S(\rho_B)- S(\rho_{AB})$
is the quantum mutual information and can be interpreted as the total correlations in \(\rho_{AB}\).

 The quantum discord from Alice to Bob $D^{\rightarrow}(\rho_{AB})$  vanishes  if and only if $\rho_{AB}$ can be decomposed in classical-quantum (CQ) form,
\begin{align} \label{CQ}
\rho_{\texttt{CQ}}&=\sum_i q_i \ket{i} \bra{i}^A \otimes \rho_{i}^B,
\end{align}
where $\{\ket{i}\}$ forms an orthonormal basis in Alice's Hilbert space, $\rho_{i}^B$ are any quantum states on Bob's Hilbert space, and $q_i$'s are probabilities. On the other hand,  quantum discord from Bob to Alice $D^{\leftarrow}(\rho_{AB})$ is defined with respect to local quantum measurements on Bob's subsystem. $D^{\leftarrow}(\rho_{AB})$ vanishes for a given $\rho_{AB}$ if and only if it is a quantum-classical (QC) state which has a form as in Eq.~(\ref{CQ}) with the subsystems $A$ and $B$   permuted. The states with zero discord both ways are called a classically correlated (CC) state, given by
\begin{equation} \label{CCIC}
\rho_{\rm CC}=\sum_{i,j}q_{ij} \ketbra{i}{i}^A \otimes \ketbra{j}{j}^B,
\end{equation}
where $\{\ket{j}\}$ forms an orthonormal basis on Bob's Hilbert space and $q_{ij}$'s are probabilities. The states that can be decomposed in the above form are incoherent with respect to the product basis $\{\ket{i}^A \otimes \ket{j}^B\}$.
Any bipartite state that is not a CC state is called a \textit{discordant state} since it has a nonzero discord in at least one way.

Entanglement as a resource in quantum information science satisfies the requirement that the quantum resource cannot be generated locally from free states \cite{CG19}. However, quantum discord as a resource does not satisfy this requirement, as quantum discord can be created locally from a CC state \cite{DVB10, SKB11, GLH+12, Gio13, BPP15}.  
Specifically, let $\Phi_A \otimes \Phi_B$ be a local operation, where $\Phi_A$ and $\Phi_B$ are completely positive trace-preserving (CPTP) maps on subsystems $A$ and $B$, respectively. Then, there are discordant states that can be created via such local operations on a CC state given by (\ref{CCIC}), with $q_{ij}= \delta_{ij} \tilde q_i$, i.e., the transformed state $(\Phi_A \otimes \Phi_B) \rho_{\rm CC}$ has a non-zero discord for a suitable choice of local operations. Such discordant states can always be written as ~\cite{GLH+12}
\begin{equation}\label{GLHd}
\rho_{AB}=\sum^{d_\lambda-1}_{\lambda=0} q_\lambda \rho^A_\lambda \otimes \rho^B_\lambda,
\end{equation}
with $d_\lambda \le d_{\min}$, where $d_{\min}:= \min\{d_A, d_B \}$~\cite{GLH+12}, for a suitable choice of local operations.

To give an example of a local creation of discord, consider the two-qubit CC state given by
\be \label{CC}
\rho_{AB}=\frac{1}{2}\left( \ketbra{00}{00}+ \ketbra{11}{11} \right).
\ee 
By applying the local operation $\Phi \otimes \Phi$, defined through $\Phi(\rho)=\ketbra{0}{0} \rho \ketbra{0}{0} + \ketbra{+}{1} \rho \ketbra{1}{+} $, where $\rho$ is a qubit state, to both qubits, we get the discordant state given by,
\be \label{LC_{A,B}}
\rho_{AB}=\frac{1}{2}\left( \ketbra{00}{00} +\ketbra{++}{++} \right). 
\ee

To verify the discordant states that cannot be generated by a local operation as mentioned above, we introduce the witness of discord~\cite{DVB10}, which is the rank of the correlation matrix of the given density matrix. This witness implies the discordant states, which do not have a local creation of the discord, as mentioned above. To introduce the witness, let us introduce the following decomposition of bipartite states~\cite{DVB10}.
 Given a bipartite state $\rho_{AB}$, it can always be decomposed as a sum of arbitrary bases of Hermitian operators $\{A_i\}$ and $\{B_i\}$ as 
\be
\rho_{AB} =\sum^{d^2_A}_{n=1}
\sum^{d^2_B}_{m=1} r_{nm} A_n \otimes B_m, 
\ee
where $d_A$ ($d_B$) is the dimension of the Hilbert space $A$
($B$). Using this representation, the correlation matrix $R = (r_{nm})$ is defined, which can be rewritten using its singular value decomposition and cast in a diagonal representation as 
\begin{equation}\label{LR}
\rho_{AB}=\sum^{L_R}_{n=1} c_n S_n \otimes F_n. 
\end{equation}
Here, $L_R$ is the rank of $R$ and quantifies how many product operators are needed to represent $\rho_{AB}$.
The value of $L_R$ can be used to witness the presence of nonzero quantum discord \cite{DVB10}. For any CQ or QC state, $L_R$ is upper bounded by the minimum dimension of the subsystems
$d_{\min}$ \cite{BPP15}. On the other hand, in general,
 the correlation rank is bounded by the square of $d_{\min}$:
$L_R \le d^2_{\min}$. Therefore, states with $L_R > d_{\min}$ will necessarily be discordant since they are neither a CQ state nor a QC state. Note that, as shown in~\cite{GLH+12},  for quantum states whose discord can be created by local operations, $L_R \le d_{\min}$. 
  
\subsection{Semi-device-independent steering} \label{SDInl}
Let us consider a steering scenario where the steering party Alice performs a set of black-box (uncharacterized) measurements labeled $x$ and obtains an outcome $a$ to produce a set of conditional states on Bob's side. This is a one-sided device-independent (1SDI) scenario, where Alice's side is device-independent, i.e., her measurements are uncharacterized, and Bob's side is trusted to perform state tomography or characterized measurements \cite{CS17}. We consider the 1SDI scenario where Bob performs a set of characterized measurements labeled $y$ on the conditional states prepared on his side and obtains an outcome $b$. In this context, a bipartite correlation $P(a,b|x,y)$, which is described by a set of conditional joint probabilities $\{p(a,b|x,y)\}_{a,x,b,y}$, is produced. 

Quantum steering is demonstrated if the correlation $P(a,b|x,y)$ does not have a \textit{local hidden variable}-\textit{local hidden state} (LHV-LHS) model \cite{WJD07}. In other words, quantum steering certifies entanglement in a 1SDI way \cite{JWD07}. Quantum steering can be witnessed through the violation of a steering inequality~\cite{CJW+09}. As an illustration of this, consider the two-setting steering scenario in which Alice performs two dichotomic black-box measurements, denoted by $A_x$, with $x=0,1$,  and Bob performs two dichotomic qubit measurements. For this scenario, the linear steering inequality is given by
\be\label{SI2}
\braket{A_0 \otimes B_0}+\braket{A_1 \otimes B_1} \le \sqrt{2},
\ee
where $B_y$, with $y=0,1$, are two of the three Pauli observables $\sigma_i$, with $i=x,y$ and $z$ and $\braket{A_x \otimes B_y}=\tr\{\rho_{AB} A_x \otimes B_y\}$~\cite{CJW+09}. As an example of the violation of this inequality, for measurements $A_0=B_0=\sigma_x$ and $A_1=B_1=\sigma_z$, the maximally entangled state of two qubits $\ket{\Phi^+}=\frac{1}{\sqrt{2}}(\ket{00}+\ket{11})$ gives rise to the maximal violation, that is, $\braket{A_0 \otimes B_0}+\braket{A_1 \otimes B_1}=2>\sqrt{2}$.

\begin{figure}
	
		\centering\includegraphics[width=5cm]{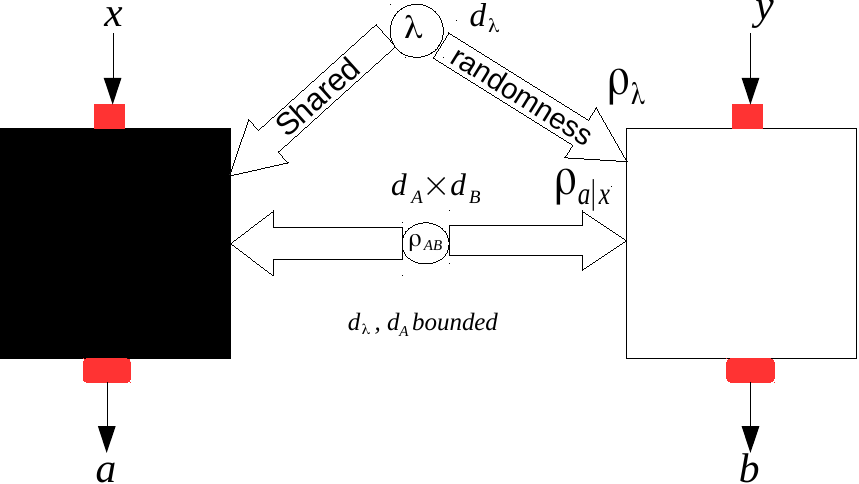}
		\caption{One-sided semi-device-independent scenario \cite{JDS_PRA23} where the local Hilbert-space dimension $d_A$ of  the shared bipartite state $\rho_{AB}$ in $\mathbb{C}^{d_A}\otimes\mathbb{C}^{d_B}$ producing the correlation $\{p(a,b|x,y)\}_{a,x,b,y}$ is also bounded. Here, the bound on $\lambda$ given by $d_{\lambda} \le d_A$ is used as a tool to witness superunsteerability; on the other hand, any amount of shared randomness $\lambda$ is not freely available if superunsteerability is used as a resource \cite{LBL+15,JD23}.  \label{Fig:superunsteer}}	
\end{figure}

Let us now consider a one-sided semi-device-independent (1SSDI) scenario, where the Hilbert space dimension of the steering party, say Alice, is trusted, as in Fig.~\ref{Fig:superunsteer}. This scenario has been considered here to define superunsteerability~\cite{DBD+18}.  Suppose that an unsteerable bipartite correlation, which has an LHV-LHS model, is produced using a quantum state in $\mathbb{C}^{d_A}\otimes\mathbb{C}^{d_B}$ for the measurements described by $\{M^A_{a|x}\}_{a,x}$ and $\{M^B_{b|y}\}_{b,y}$, which are a 
POVM, i.e., for each $x$,  $M^A_{a|x} \ge 0$ and $\sum_a M^A_{a|x} = \openone_{d_A}$, where $\openone_{d_A}$ is the identity operator acting on $\mathbb{C}^{d_A}$ and similar conditions hold for Bob's measurements for each $y$.  Then, the LHV-LHS model of the correlation can be decomposed as follows:
\be
p(a,b|x,y)= \sum^{d_\lambda-1}_{\lambda=0} p_{\lambda} p(a|x,\lambda) p(b|y; \rho_\lambda) \quad \forall a,x,b,y, \label{LHV-LHS}
\ee
where  $d_\lambda$ denotes the dimension of the finite shared randomness $\lambda$ used to simulate the correlation. 
$\{p(a|x,\lambda)\}_{a,x}$ in Eq. (\ref{LHV-LHS}) is the set of arbitrary probability distributions $p(a|x,\lambda)$ conditioned on shared randomness/hidden variable $\lambda$ occurring with probability $p_{\lambda}$; $\sum^{d_\lambda-1}_{\lambda=0} p_{\lambda} =1$. On the other hand, $\{p(b|y; \rho_\lambda) \}_{b,y}$ is the set of quantum probability distributions $p(b|y; \rho_\lambda)=\tr\{M^B_{b|y}\rho_\lambda\}$, arising from some LHS $\rho_\lambda$,  which satisfies $\sum_{a,\lambda} p_\lambda p(a|x,\lambda) \rho_\lambda=\rho_B$, where $\rho_B$ is the reduced state of $\rho_{AB}$ producing the correlation. 
Superunsteerability is defined as follows.
\begin{definition}
Suppose an unsteerable correlation, which has an LHV-LHS model as in Eq. (\ref{LHV-LHS}), is produced in a 1SSDI scenario, as in Fig.~\ref{Fig:superunsteer}. Then, superunsteerability is demonstrated if and only if $d_\lambda$ in Eq. (\ref{LHV-LHS})  satisfies $d_\lambda >  d_A$ for all such LHV-LHS models. 
\end{definition}
Superunsteerability has been invoked to provide an operational characterization of the quantumness of unsteerable bipartite states~\cite{Jeb14, DBD+18, JD23}. Here, the quantumness arises because superunsteerability requires quantum discord in the given unsteerable state.

More specifically, superunsteerability witnesses 
SDI steerability of discordant states~\cite{JDS_PRA23,JDS24}.
 SDI steering can be defined as follows.
\begin{definition}\label{SDISteerDef}
In a 1SSDI scenario, as in Fig.~\ref{Fig:superunsteer}, SDI steering is demonstrated if a superunsteerable correlation, which implies superunsteerability,  or a steerable correlation, which does not have an LHV-LHS model, is produced. 
\end{definition}
As an illustration, consider the Werner state given by
\be \label{Wers}
\rho_W=p \ketbra{\Phi^+}{\Phi^+} +(1-p) \frac{\openone_2 \otimes \openone_2}{4},
\ee 
where  $0 \le p \le 1$. $\rho_W$ is entangled if and only if $p > 1/3$; at the same time, it has non-zero discord for any $p>0$. On the other hand, in the two-setting steering scenario, $\rho_W$ has steerability for $p>1/\sqrt{2}$, while it has superunsteerability for $0 < p \le 1/\sqrt{2}$ \cite{DDJ+18,JD23}. Thus, $\rho_W$ has SDI steerability for any $p>0$.


 In the context of a 1SSDI scenario, as shown in Fig.~\ref{Fig:superunsteer}, we have the following assumption.
\begin{observation}\label{FOSR}
 In the given 1SSDI scenario, to observe SDI steerability of discordant states that are not steerable in the 1SDI context, shared randomness $\lambda$ is not freely available.
\end{observation}
This assumption follows because the set of nonsuperunsteerable correlations is non-convex~\cite{DBD+18,JD23,JDS_PRA23}. If shared randomness is freely available in the context of a 1SSDI scenario, then the correlations that have an LHV-LHS model do not certify SDI steerability. Thus, entanglement is required to demonstrate SDI steering if shared randomness is freely available.

\section{Global coherence}\label{GCdefRT}
The quantum coherence of a state represents superposition with respect to a fixed basis. Different measures of quantum coherence have been studied extensively in the literature \cite{SAP17}.
For a single system, fixing a reference basis $\{\ket{i}\}$, a state $\rho$ of the system is said to have vanishing coherence if it is diagonal in the reference basis, i.e., 
\be 
\rho =\sum_i p_i \ketbra{i}{i},
\ee
with probabilities $p_i$.
In~\cite{BCP14}, the relative entropy of coherence was introduced as a coherence measure, analogous to the relative entropy of entanglement, which quantifies entanglement \cite{VP98}. 
 The quantum relative entropy of a state $\rho$ with respect to a state $\sigma$ is defined as
$S(\rho\vert\vert\sigma) \equiv -\tr{\rho\log_2\sigma}-S(\rho)$, where
$S(\rho)=-\tr{\rho\log_2\rho}$.
The relative entropy of coherence ${\cal C}^r(\rho)$ of a state $\rho$ with respect to a basis is then defined as 
\begin{equation}\label{eqc}
{\cal C}^r(\rho) \equiv \min_{\sigma \in {\cal I}} S(\rho\vert\vert\sigma),
\end{equation}
where $ {\cal I}$ is the set of totally incoherent (diagonal) states in that basis.

 Moving to the bipartite systems, fixing a reference product basis $\{\ket{i}^A \otimes \ket{j}^B\}$,  a bipartite state is incoherent (IC) if it can be written as follows \cite{BCA15,SSD+15}: 
\be \label{IC}
\rho^{IC}_{AB}=\sum_k p_k \tau^A_k \otimes \upsilon^B_k, 
\ee
where $p_k$'s are probabilities,  $\tau^A_k$ and $\upsilon^B_k$ are incoherent states on subsystems $A$ and $B$, respectively,
 i.e., $\tau^A_k=\sum_i p'_{ik} \ketbra{i}{i}$ and $\upsilon^B_k=\sum_j p''_{jk} \ketbra{j}{j}$ for probabilities  $p'_{ik}$ and $p''_{jk}$. Otherwise, it is said to have \textit{bipartite coherence}. The relative entropy of coherence can also be used to quantify bipartite coherence \cite{GZ18}. 
 The set of incoherent states forms a subset of all separable states.  As in the case of incoherent operations, which are free operations in a resource theory of coherence for single systems \cite{BCP14}, incoherent operations on bipartite states, described by a CPTP map, transform any incoherent bipartite state into another incoherent bipartite state. 
 These incoherent operations are defined globally, and they do not define the free operations in the resource theory of coherence in bipartite scenarios. The free operations of coherence in bipartite scenarios were studied in~\cite{SRB+17}.

 Quantumness due to a nonzero discord has been studied using the notion of quantum coherence \cite{ABC16,HHW+18}.  Any incoherent bipartite state has zero discord in both ways; on the other hand, a nonzero discord in one way or both ways always implies bipartite coherence. More specifically, any discordant state has bipartite coherence in a basis-free form \cite{YXG+15},  because it is not a CC state (\ref{CCIC}) with respect to any product basis.

As argued in~\cite{GLH+12}, the quantumness of any discordant state that has a decomposition as in Eq.~(\ref{GLHd})  is due to local coherence in the subsystem $A$ or $B$ or both simultaneously.
We now obtain the following observation to identify all bipartite states that have bipartite coherence locally.
\begin{observation}
All discordant states in $\mathbb{C}^{d_A}\otimes\mathbb{C}^{d_B}$, which have a decomposition as in Eq.~(\ref{GLHd}) with $d_\lambda \le d_{\min}$, form a strict subset of all discordant states in $\mathbb{C}^{d_A}\otimes\mathbb{C}^{d_B}$ with $L_R \le d_{\min}$. 
\end{observation}
 To prove the above observation, note that if $d_A \ne d_B$, the set of separable states in $\mathbb{C}^{d_A}\otimes\mathbb{C}^{d_B}$, which have a decomposition as in Eq.~(\ref{GLHd}), forms a subset of all separable states in $\mathbb{C}^{d_A}\otimes\mathbb{C}^{d_B}$ with $L_R \le d_{\min}$ \cite{GLH+12}. This can be seen as follows. If $d_A=d_B=d$, all states with $L_R \le d_{\min}$ can be decomposed as in Eq.~(\ref{GLHd}) with $d_{\min}=d$. 
On the other hand, if  $d_A < d_B$ or  $d_A > d_B$, there are discordant states with $L_R \le d_{\min}$ that do not have a local creation of discord~ \cite{GLH+12}. We will give an example of such a state. To this end, denote by
    \begin{equation} \label{eq:desc_sep_1}
    \ket{\theta,\phi} := \cos\left(\frac{\theta}{2}\right) \ket{0}
                        + \exp(i \phi)\sin\left(\frac{\theta}{2}\right) \ket{1} \,,
    \end{equation}
an arbitrary pure state in $\mathbb{C}^2$, then the set ${\mathcal{W}=\{\ket{0,0},\ket{\frac{2\pi}{3},0},\ket{\frac{2\pi}{3},\pi}\}}$ can be used to define a QC state in $\mathbb{C}^{2}\otimes\mathbb{C}^{3}$ as follows:
    \begin{equation} \label{QCng}
    \rho_{AB} = \frac{1}{3}\sum_{\lambda=0}^2 W_\lambda \otimes \ketbra{\lambda}{\lambda}   \,,
    \end{equation}
    with ${W_\lambda=\ket{w_\lambda}\bra{w_\lambda}}$ and ${\ket{w_\lambda}\in\mathcal{W}}$. The above discordant state does not have a local creation of discord \cite{BPP15}, yet it has $L_R =2$ since it has $D^{\leftarrow}(\rho_{AB})=0$. 

 We now make the following observation.
 \begin{observation} \label{CQQCnSsep}
     All discordant states in $\mathbb{C}^{d_A}\otimes\mathbb{C}^{d_B}$ with $L_R \le d_{\min}$,  which do not have a decomposition as in Eq. (\ref{GLHd}) with $d_\lambda \le d_{\min}$, is a CQ state if $d_A > d_B$ or a QC state if $d_A < d_B$. For such a CQ (QC) state, the dimension of the support of its reduced state $\rho_A$ ($\rho_B$) is greater than that of $\rho_B$ ($\rho_A$). 
 \end{observation}
 This observation follows because if the state is neither a CQ state nor a QC state, it has a local creation of discord as shown in~\cite{GLH+12}  or $L_R > d_{\min}$. On the other hand, if the dimension of the support of the reduced state $\rho_A$ ($\rho_B$) of the CQ (QC) state mentioned in the above observation is not greater than that of $\rho_B$ ($\rho_A$), it has a decomposition as in Eq. (\ref{GLHd}) with $d_\lambda \le d_{\min}$.

 We have now obtained the following observation.
 \begin{observation}\label{lcs}
     For any given separable state in $\mathbb{C}^{d_A}\otimes\mathbb{C}^{d_B}$ with $L_R \le d_{\min}$,  its bipartite coherence is locally present in a subsystem or in both subsystems. 
 \end{observation} 
 The above observation follows because the bipartite coherence of any CQ or QC state is contained locally in the subsystem $B$ or $A$, respectively.
 On the other hand, all other states with $L_R \le d_{\min}$, which have nonzero discord in both ways, can be created by a local operation on a CC state \cite{GLH+12}.
Using Obs.~\ref{lcs}, we can define the global coherence of bipartite states as follows.
\begin{definition} \label{DefGC}
 A bipartite state, $\rho_{AB}$ in $\mathbb{C}^{d_A}\otimes\mathbb{C}^{d_B}$, has global coherence if and only if it is neither a CQ state nor a QC state and its discord cannot be generated by a local operation of the form $\Phi_A \otimes \Phi_B$ on a CC state given by Eq.~ (\ref{CCIC}) with $q_{ij}= \delta_{ij} \tilde q_i$. 
\end{definition}
Note that all entangled states have global coherence because they are not separable to have bipartite coherence locally to any of the subsystems or both subsystems. At the same time, global coherence goes beyond entanglement. For example, the Werner state $\rho_W$ given by Eq. (\ref{Wers}), which is entangled if and only if $p > 1/3$,  has $L_R > 2$ for any $p>0$. Thus, $\rho_W$ has global coherence for any $p>0$.

In the case of global coherence as a resource, the free states are all incoherent bipartite states and bipartite coherent states whose coherence is present locally, i.e. states with $L_R \le d_{\min}$, and the free operations are given by any local operation of the form $\Phi_A \otimes \Phi_B$, where $\Phi$ is any CPTP map because such operations cannot create the resource of global coherence from any free state.  To illustrate how the free operations of global coherence cannot generate a resource state from a free state, we have already seen an example in which, applying the local operation $\Phi \otimes \Phi$, to both qubits of the CC state given by Eq.~(\ref{CC}), we get the discordant state given by Eq.~(\ref{LC_{A,B}}), which, however, does not have global coherence.
 On the other hand, a globally coherent state can be generated locally from the CC state (\ref{CC}) by applying a local operation of the form given by $\sum_\lambda p_\lambda  \Phi_\lambda \otimes \Phi_\lambda  $ (each $\Phi_\lambda$ is a CPTP map and $p_\lambda \ge 0$, with $\sum_\lambda p_\lambda =1$).
  Note that such local operations and shared randomness are free operations in the case of entanglement as a resource. For comparison of free operations of global coherence in quantum discord with that of entanglement theory, please refer to Eq.~($18$) in~\cite{CG19}. For the resource theory of discord, free operations have also been studied in~\cite{HFZ+12, LHL17}. These free operations are more restrictive than those of global coherence, since the former should not create quantum discord locally. Like in the case of quantum discord, the free states of global coherence of quantum discord are also nonconvex.

Based on any contractive pseudo-distance measure \cite{Bengtsson_Życzkowski_2017}, one can construct a measure of global coherence. Given a bipartite state $\rho_{AB}$, a contractive pseudo-distance measure of global coherence  ${\cal G}^\delta(\rho_{AB})$  is defined as 
  \be
  {\cal G}^\delta(\rho_{AB})=\min_{\chi_{AB} \in \mathcal{F}_{G}} \delta (\rho_{AB},\chi_{AB}),
  \ee
  where  $\mathcal{F}_{G}$ denotes the set of free states in the resource theory of global coherence and  $\delta$ is the pseudo-distance measure required to be contractive under quantum operations, i.e.,
  \be \label{Dmeaure}
   \delta (\Phi [\rho],\Phi [\zeta]) \le \delta(\rho,\zeta),
  \ee
where $\Phi$ is any CPTP map. This implies that ${\cal G}^\delta$ does not increase under any local operation of the form $\Phi_A \otimes \Phi_B$, which is a free operation in the case of global coherence as a resource. 

We proceed to study the relationship between the resource frameworks of global coherence of bipartite and single-partite states.
To this end, we consider the protocol for converting the coherence of a single system into discord of a bipartite system under a bipartite incoherent operation, as studied in~\cite{MYG+16}. 
Suppose a coherent state $\rho_A$ is coupled to an initially incoherent ancilla $\tau_B$ and an incoherent operation $\Lambda_{IC}$ is performed on the joint state $\rho_{AB}=\rho_A \otimes \tau_B$. Then, there exists a bipartite incoherent operation $\Lambda_{IC}$ for which coherence can be converted into entanglement as shown in~\cite{SSD+15}. At the same time, discord without entanglement can also be generated~\cite{MYG+16}, as in the case of the conversion of coherence into entanglement using a bipartite incoherent operation. 
For instance, consider that the  bipartite incoherent operation $\varepsilon(\rho_{AB})=p U_{\mathrm{CX}}\rho_{AB} U_{\mathrm{CX}} + (1-p) \frac{\openone_4}{4} $, where $U_{\mathrm{CX}}$ is the controlled-NOT gate with $U_{\mathrm{CX}}(\ket{i}\otimes\ket{j})=\ket{i}\otimes\ket{i\oplus j}$. If this incoherent operation is performed on the joint state $\rho_{AB}=\ketbra{+}{+} \otimes \ketbra{0}{0}$, then the joint state is transformed into a Werner state $\rho_W$. Note that in this context, discord without entanglement is generated by tuning $ 0<p \le 1/3$.

There is also a bound on the amount of discord generated in the aforementioned protocol to convert coherence into discord~\cite{MYG+16}. To state this bound,
let ${\cal C}^\delta$ denote a contractive pseudo-distance measure of coherence of a single-partite state $\rho$ given by
  \be
  {\cal C}^\delta(\rho)=\min_{\sigma \in \cal {I}} \delta (\rho,\sigma),
  \ee
such as in Eq.~(\ref{eqc}), and ${\cal D}^\delta$ denote the  contractive pseudo-distance measure of discord, i.e.,
\be
  {\cal D}^\delta(\rho_{AB})=\min_{\sigma_{AB} \in \mathcal{C}} \delta (\rho_{AB},\sigma_{AB}),
  \ee
  where  $\mathcal{C}$ denotes the set of CC states given by Eq. (\ref{CCIC}).
Then, the amount of discord generated satisfies the following inequality:
\be \label{sc-gcd}
{\cal D}^\delta(\Lambda_{IC}(\rho_A \otimes \tau_B)) \le {\cal C}^\delta(\rho_A),
\ee  
as shown in~\cite{MYG+16}.

Suppose $\Lambda_{IC}(\rho_A \otimes \tau_B)$ is a separable state and has ${\cal D}^\delta>0$. Then this discordant state necessarily has ${\cal G}^\delta>0$. This follows because the coherence of the subsystem must be consumed to create the nonzero discord, as was also demonstrated in Result $2$ in~\cite{MYG+16}. Next, note that 
\be \label{GleD}
{\cal G}^\delta(\rho_{AB})\le {\cal D}^\delta(\rho_{AB}),
\ee
which follows due to the contractivity of $\delta$, and there are discordant states for which ${\cal D}^\delta(\rho_{AB})>0$, but ${\cal G}^\delta(\rho_{AB})=0$. 
Combining Eqs.~(\ref{sc-gcd}) and (\ref{GleD}), we have now obtained the following relationship between ${\cal G}^\delta$ and ${\cal C}^\delta$.
\begin{observation}
  For any contractive distance $\delta$, the amount of global coherence created between a state $\rho_A$ and an incoherent ancilla $\tau_B$  by an incoherent operation $\Lambda_{IC}$ satisfies the inequality:
\be \label{sc-gc}
{\cal G}^\delta(\Lambda_{IC}(\rho_A \otimes \tau_B)) \le {\cal C}^\delta(\rho_A).
\ee  
\end{observation}
Here, we have derived the inequality in Eq. (\ref{sc-gc}) to establish a strong link between the resource frameworks of global coherence of bipartite and single-partite states, extending the result of~\cite{SSD+15}, which demonstrated a strong link between the resource frameworks of entanglement and coherence in single systems. The inequality in Eq. (\ref{sc-gcd}) is only quantitative to capture genuinely quantum correlation generated; on the other hand, our inequality in Eq. (\ref{sc-gc}) captures both qualitative and quantitative aspects of genuinely quantum correlation generated in the protocol.

\section{Operational tasks requiring global coherence} \label{OTGR}

In the following, we explore the relationship between global coherence and  
SDI steerability in 1SSDI scenarios, as in Fig.~\ref{Fig:superunsteer}. We consider these steering tasks because shared randomness is not a free resource for such operational tasks, as mentioned in Obs.~\ref{FOSR}; at the same time, the free operations of global coherence refer to any local operation without shared randomness. We will also explore the role of global coherence in the task of remote state preparation using quantum discord ~\cite{DLM+12} because the task has to use only local operations freely without shared randomness~\cite{Gio13}.  

\subsection{Semi-device-independent (SDI) steering}
Here we demonstrate that global coherence is required to demonstrate SDI steerability in the context of a 1SSDI scenario, as in Fig.~\ref{Fig:superunsteer}. To this end, in the following, we prove two lemmas on superunsteerability to obtain our main result.
 \begin{lem} \label{lem1}
If a separable state in $\mathbb{C}^{d_A}\otimes\mathbb{C}^{d_B}$ has a decomposition as in Eq.~(\ref{GLHd}), it cannot be used to demonstrate superunsteerability.
\end{lem}
\begin{proof}
Consider any correlation arising from a separable state of the form (\ref{GLHd}) in the context of a 1SSDI scenario, as in Fig.~\ref{Fig:superunsteer}. It is readily seen that any such correlation admits the following decomposition:
\begin{equation} \label{LHVLR<=dmin}
p(a,b|x,y)=\sum^{d_\lambda-1}_{\lambda=0} p_{\lambda} p(a|x,\lambda) p(b|y, \lambda),
\end{equation}
with $d_\lambda \le d_{\min}$. Here $p(a|x,\lambda)=\tr \left\{M^A_{a|x} \rho^A_\lambda \right\}$ and $p(b|y,\lambda)=\tr \left\{M^B_{b|y} \rho^B_\lambda \right\}$. From the above decomposition, it then follows that any separable state, which has a decomposition as in Eq. (\ref{GLHd}), cannot be used to demonstrate superunsteerability, as it has an LHV-LHS. 
\end{proof}

As we have noted before, the set of separable states in $\mathbb{C}^{d_A}\otimes\mathbb{C}^{d_B}$, which have a decomposition as in Eq.~(\ref{GLHd}), forms a subset of all separable states in $\mathbb{C}^{d_A}\otimes\mathbb{C}^{d_B}$ with $L_R \le d_{\min}$  \cite{GLH+12,BPP15}. Thus, there are discordant states that do not have a decomposition as in Eq. (\ref{GLHd}); at the same time, they do not have global coherence. Such states have the form of a CQ or QC state. For any of these states, it is not immediately apparent from their CQ or QC form whether they can be used to demonstrate superunsteerability or not. For example,  consider any correlation arising from the QC state given by Eq.~(\ref{QCng}).  It can be seen that there exists a correlation arising from such a QC state that can be decomposed as in Eq.~(\ref{LHVLR<=dmin}), but with $d_\lambda = 3$, which may not be the minimal dimension required to simulate the correlation. Thus, from this decomposition of the correlation, it is impossible to immediately conclude whether the correlation is superunsteerable or not.  To demonstrate that any CQ or QC state cannot be used to demonstrate superunsteerability, we make the following observation.
\begin{observation}\label{ObsVed}
  Suppose that any discordant state in $\mathbb{C}^{d_A}\otimes\mathbb{C}^{d_B}$ has no global coherence, i.e., $L_R \le d_{\min}$. Its nonzero discord is due to a nonvanishing classical correlation $C^\rightarrow$ as given in Eq.~(\ref{QDdef}) or $C^\leftarrow$, i.e., the presence of correlation in only a single basis, assisted with local coherence in a subsystem or both subsystems. \end{observation}
The above observation follows because such discordant states can be produced from a CC state (\ref{CCIC}) by locally generating coherence, or they contain coherence locally as a CQ or QC state. Therefore, the discord of these states implies the presence of correlation in only a single basis. 
Whereas, for discordant states with $L_R > d_{\min}$, their nonzero discord signifies a form of genuinely quantum correlation.  Such a form of genuinely quantum correlation indicates the presence of correlations in more than one basis, just as in the case of entangled states. 
As an illustration of correlation in more than one basis, note that a Bell state of a two-qubit system exhibits correlations in both bases of the Pauli matrices $\sigma_x$ and $\sigma_z$, implying steerability; at the same time, all discordant states that belong to the two-qubit Bell-diagonal states have correlations in more than one basis \cite{WMC+14}, implying that they have SDI steerability \cite{JDS24}. 

We now proceed to obtain the following lemma using Obs.~\ref{ObsVed}.
\begin{lem} \label{lem2}
Suppose that any given discordant state in $\mathbb{C}^{d_A}\otimes\mathbb{C}^{d_B}$  does not have a decomposition as in Eq.~(\ref{GLHd}), but it has $L_R \le d_{\min}$.  Then, it cannot be used to demonstrate superunsteerability.
\end{lem}
\begin{proof} 
As noted in Obs.~\ref{CQQCnSsep}, any given discordant state in $\mathbb{C}^{d_A}\otimes\mathbb{C}^{d_B}$ that  does not have a decomposition as in Eq.~(\ref{GLHd}) is a CQ state if $d_A > d_B$ or a QC state if $d_A < d_B$.  
Since any CQ or QC state has no global coherence and has only correlation in a single basis, as noted in Obs.~\ref{ObsVed}, any correlation arising from this state can be reproduced by a suitable CC state. We show this by first considering any QC state in $\mathbb{C}^{d_A}\otimes\mathbb{C}^{d_B}$, with $d_A < d_B$, that is discordant  given by 
\begin{equation}\label{QCDcc}
\rho_{\rm QC}=\sum_{i} q_{i} \rho^A_i  \otimes \ketbra{i}{i}^B,
\end{equation}
as mentioned in Obs.~\ref{CQQCnSsep}.
The QC state given by Eq. (\ref{QCDcc}) is not a CC state because it has a nonzero discord. An example of such a QC state is given by Eq.~(\ref{QCng}).  

The QC state (\ref{QCDcc}) can be transformed into a CC state, while retaining the same amount of the classical correlation $C^{\rightarrow}(\rho_{\rm QC})$, by applying a suitable decohering map on subsystem $A$. To see this, consider the decohering map,  
$\Phi_A=\sum_iK_i(\cdot)K^\dag_i $, with Kraus operators $K_i$ satisfying $\sum_i K^\dag_i K_i=\openone_{d_A}$ and $K_i=\ketbra{\chi_i}{\phi_i}$, where $\{\ket{\phi_i}\}$ is a basis that gives rise to the classical correlation $C^{\rightarrow}(\rho_{\rm QC})$ with projective measurement and $\{\ket{\chi_i}\}$ is any orthogonal basis. Applying this map on $\rho_{\rm QC}$,  we get  the  state, $\rho_{\rm CC}$, as follows: 
\begin{equation} \label{IC_A}
\rho_{\rm CC}=\sum_i (K_i \otimes \openone_{d_B}) \rho_{\rm QC} (K^\dag_i \otimes \openone_{d_B}),
\end{equation}
which has the same amount of the classical correlation as that of $\rho_{\rm QC}$, and at the same time, it has vanishing quantum discord from $A$ to $B$ because it is the CC state.

Now, consider any correlation that arises from $\rho_{\rm QC}$ given by 
\begin{align} \label{BLC_A,B}
p(a,b|x,y)=\tr \left\{\rho_{\rm QC} M^A_{a|x} \otimes M^B_{b|y}  \right\}, \forall a,x,b,y,
\end{align}
 for the POVM on Alice's and Bob's sides given by  $\{M^A_{a|x}\}_{a,x}$ and $\{M^B_{b|y}\}_{b,y}$, respectively. Any such correlation can be reproduced by the state $\rho_{\rm CC}$ in Eq.~(\ref{IC_A}) for the new measurement operators on Alice's side given by $\{\sum_i K^\dag_i M^A_{a|x} K_i\}_{a,x}$ and the same measurements on Bob's side. To see this, for such measurements, consider the correlation, $\{p'(a,b|x,y)\}_{a,x,b,y}$, arising from the CC state, $\rho_{\rm CC}$ in Eq. (\ref{IC_A}), given by
\begin{align}
p'(a,b|x,y)&=\tr \left\{\rho_{\rm CC} \sum_i K^\dag_i M^A_{a|x} K_i \otimes M^B_{b|y}    \right\}, 
\end{align}
for all $a,x,b,y$. Now, substituting $\rho_{\rm CC}$ in Eq. (\ref{IC_A}) in the above equation, we obtain  
\begin{multline}
p'(a,b|x,y)=\tr \Bigg\{  \sum_{i} (K_{i} \otimes \openone_{d_B}) \rho_{QC} (K^\dag_{i} \otimes \openone_{d_B})\\
\times \sum_{i'} K^\dag_{i'} M^A_{a|x} K_{i'} \otimes M^B_{b|y}   \Bigg\}. 
\end{multline}
 Now, using the invariance of the trace under cyclic permutations of the operators inside the trace, we have 
  \begin{multline}
 p'(a,b|x,y)\\= \tr \left\{\rho_{QC} \sum_{i,i'}K^\dag_{i'}  K_i  M^A_{a|x} K^\dag_i K_{i'}  \otimes M^B_{b|y}    \right\}.
 \end{multline}
 Note that $K_i$ as defined above satisfy the orthogonality condition $K^\dag_i K_{i'}=K_i K^\dag_{i'}=0$, for $i\ne i'$. Using this condition and  $\sum _iK^\dag_i K_{i}=\openone_{d_A}$  in the above equation, we obtain, 
 \begin{align}
p'(a,b|x,y)=&\tr \left\{\rho_{\rm QC} M^A_{a|x} \otimes M^B_{b|y}  \right\}, 
\end{align}
for all $a,x,b,y$, which is the same correlation as in Eq.~(\ref{BLC_A,B}). This implies that the correlation has an LHV-LHS model with $d_\lambda \le d_{\min}$, as it can be reproduced by the CC state. Similarly, this statement also holds in the case of the CQ states, which do not have a decomposition as in Eq.~(\ref{GLHd}).
\end{proof}

We have now obtained the following main result of this paper. 
\begin{thm}\label{thm}
For any given bipartite state in $\mathbb{C}^{d_A}\otimes\mathbb{C}^{d_B}$, global coherence is required to demonstrate SDI steering.
\end{thm}
\begin{proof}
To prove the statement of this theorem, we first note that all bipartite states that exhibit SDI steering can be classified into the following three categories. 
\begin{enumerate}
    \item All states that exhibit standard steering. Such states are necessarily entangled~\cite{WJD07,JWD07} and have SDI steerability since standard steering also implies SDI steering by Definition~\ref{SDISteerDef}. 
    \item There are entangled states that do have standard steerability, but exhibit SDI steering~\cite{DBD+18,JDS_PRA23}.
    \item There are separable states that exhibit SDI steering~\cite{DBD+18,JDS_PRA23}.
\end{enumerate}
We proceed to prove the theorem using this classification.
Because entanglement implies global coherence, in the case of the first two categories of states, global coherence is necessary for SDI steering.
 On the other hand, Lemma~\ref{lem1} and Lemma~\ref{lem2} imply that any separable state with $L_R \le d_{\min}$ cannot be used to demonstrate superunsteerability. Thus, global coherence is required to demonstrate SDI steering using separable states, since $L_R>d_{\min}$ implies global coherence by Definition \ref{DefGC} and superunsteerability witnesses SDI steering by Definition~\ref{SDISteerDef}. Since we have now arrived at the conclusion that global coherence is required for SDI steering in all of the three categories of states that exhibit SDI steering, we have proved the statement of the theorem.      
\end{proof}


\subsection{Remote state preparation (RSP)}  
We consider the quantum information protocol, 
RSP~\cite{BDS+01,Lo00,Pat00} using two-qubit systems. In this context,  in~\cite{DLM+12}, it was proposed that quantum discord captures the resourcefulness of two-qubit states for RSP. We review this result as follows.
If Alice and Bob share  a maximally entangled state of a two-qubit system, as a resource in the RSP protocol, then RSP can be achieved perfectly \cite{BBC+93}.  Instead, if they use a more general resource state, such as a mixed state of two-qubit systems, RSP could occur imperfectly. In this context, to quantify the performance of the RSP protocol, a figure of merit, called RSP-fidelity, was introduced in~\cite{DLM+12}. Let Alice use a shared two-qubit state  to remotely  prepare a state with the Bloch vector $\vec{s}$ in the plane orthogonal to the direction $\vec{\beta}$ on Bob's side.  Since Alice uses a mixed resource state, she remotely prepares a state with the Bloch vector $\vec{r}$ instead of  $\vec{s}$. To define the RSP-fidelity, the authors of~\cite{DLM+12} first defined the pay-off function as ${\cal P}:=(\vec{r} \cdot \vec{s})^2$. Next, for the given $\vec{s}$ and $\vec{\beta}$, considering the optimal payoff  ${\cal P}_{\max}$, which is the payoff optimized overall local measurements on Alice's side, the RSP-fidelity $\mathcal{F}$ is finally defined to be $\mathcal{F}(\rho_{AB}):=\min_{\vec{\beta}} {\cal P}^{\rm ave}_{\max}$, where ${\cal P}^{\rm ave}_{\max}$ is the optimal payoff averaged over all states $\vec{s}$ in the given plane of the Bloch sphere, and  ${\cal P}^{\rm ave}_{\max}$ minimized over all directions $\vec{\beta}$.

Any two-qubit state, up to local unitary transformations, can be expressed as~\cite{Luo08}:
\begin{multline}\label{can02q}
\zeta_{AB}=\frac{1}{4}\Bigg(\openone_2 \otimes \openone_2 + \vec{a} \cdot \vec{\sigma} \otimes \openone_2 
+ \openone_2 \otimes \vec{b} \cdot \vec{\sigma}  
+\sum^3_{i=1} c_i \sigma_i \otimes \sigma_i \Bigg),  
\end{multline}
where  $\openone_2$ is the $2 \times 2$ identity matrix, $\sigma_i$, with $i=1,2$ and $3$ are the three Pauli matrices, respectively,  $\vec{\sigma}$ is the vector of these Pauli matrices and $\{\vec{a},\vec{b},\vec{c}\} \in \mathbf{R}^3$ are vectors with norm less than or equal to unity satisfying $\vec{a}^2+\vec{b}^2+\vec{c}^2\le3$.
For two-qubit states, $\zeta_{AB}$, as given in Eq. (\ref{can02q}), with $c^2_1\ge c^2_2 \ge c^3_3$, the RSP-fidelity $
\mathcal{F}(\zeta_{AB})$ was evaluated to be of the form \cite{DLM+12},
\begin{align}
\mathcal{F}(\zeta_{AB})=\frac{1}{2}(c^2_2+c^2_3).
\end{align}
In~\cite{DLM+12}, a direct comparison for zero discord states, which have $c_2=c_3=0$, with the above RSP-fidelity suggested that quantum discord would be a necessary and sufficient resource for the RSP protocol. Interestingly, it was also shown that the above RSP-fidelity in the case of the Bell-diagonal states is equal to the geometric measure of quantum discord \cite{DVB10}.
However, in~\cite{Gio13}, it was shown that quantum discord is not the precise resource for the RSP protocol studied in~\cite{DLM+12}. 
 
 We proceed to obtain another main result of the paper in identifying the precise resource for the RSP protocol using separable two-qubit states. To this end, we consider a specific decomposition of two-qubit states. Before that, we note that as shown in~ \cite{LS98}, any two-qubit state can be written as a convex combination of a pure entangled state $\ket{\psi_e}$ and a separable state $\rho_{\texttt{sep}}$ as follows:
\be \label{LSe}
\rho_{AB}=p \ketbra{\psi_e}{\psi_e}+(1-p) \rho_{\texttt{sep}}.
\ee
In the above decomposition, the weight of the entangled state $p$ \textit{minimized} over all possible decompositions provides a \textit{measure of entanglement} \cite{LS98}.

On the other hand, for our purpose, we consider a decomposition as in Eq.~(\ref{LSe}) in which the weight of the entangled state $p$ \textit{was maximized} over all possible decompositions such that it provides \textit{a measure of global coherence}. For the two-qubit Werner states as given by Eq.~(\ref{Wers}),
the weight $p$ in the decomposition gives the maximal entangled part to quantify global coherence. This maximal weight of an entangled state as a measure of global coherence is analogous to the quantifier \textit{Schr\"{o}dinger strength}~\cite{JD23}, as reviewed in App. \ref{AppSS}, which provides
a measure of SDI steerability.

Consider two-qubit Bell diagonal states $\tau_{AB}$ that are the states in Eq.~(\ref{can02q}) with local Bloch vectors $\vec{a}$ and $\vec{b}$ being zero vectors, i.e., all two-qubit states with maximally mixed marginals.
For two-qubit Bell-diagonal states $\tau_{AB}$ which include the Werner state given by Eq.~(\ref{Wers}), the Schr\"{o}dinger strength in the two-setting scenario $SS_{2}$ has been known~\cite{JD23} and is given by Eq.~(\ref{SSBD}) in Appendix~A. To illustrate the idea of this quantifier $SS_{2}$, 
the decomposition of the correlation that implies this Schr\"{o}dinger strength of the Werner state, i.e.,  $SS_{2}(\rho_W)=p$, is given by 
 \begin{equation}
P(a,b|x,y)=p P_S(a,b|x,y) + (1-p) P_{US}(a,b|x,y),
\end{equation}
where $P_S(a,b|x,y)$ is an extremal steering correlation arising from $\ket{\Phi^+}$ that violates the two-setting steering inequality (\ref{SI2}) maximally and $P_{US}(a,b|x,y)$ is an unsteerable correlation arising from the product state in Eq.~(\ref{Wers}), i.e., $\frac{\openone_2 \otimes \openone_2}{4}$. Note that in the decomposition of the Werner state given by Eq.~(\ref{Wers}), the separable state is the product state and the weight of the entangled state also gives the  \textit{Schr\"{o}dinger strength} of the state in the two-setting scenario. 

Next, for the two-qubit states beyond the Werner states, we identify the separable state for the \textit{optimal decomposition} that maximizes the entangled part in Eq.~(\ref{LSe}) to quantify global coherence. To this end,  we define a quantity called \textit{Mermin strength}, $\Gamma$, in the context of the scenario as in Fig. \ref{Fig:superunsteer}, where $a,x,b,y \in \{0,1\}$. This quantity is defined in terms of the covariance of the Mermin functions \footnote{Such a covariance Mermin function is inspired by covariance Bell inequalities studied in~\cite{PHC+17}.} that was used to define a quantity called \textit{Mermin discord} in~\cite{Jeb14}.  
Consider the covariance,  $\texttt{cov}(A_x,B_y)$, of the observables on Alice and Bob's side $A_x$ and $B_y$ defined as $A_x:=M^A_{0|x}-M^A_{1|x}$ and $B_y:=M^B_{0|y}-M^B_{1|y}$, respectively, given by
 \be
 \texttt{cov}(A_x,B_y)=\braket{A_xB_y} -\braket{A_x}\braket{B_y},
 \ee
 where $\braket{A_xB_y}$ and $\braket{A_x}$, and $\braket{B_y}$
 are joint and marginal expectation values, respectively.

To define $\Gamma$, the absolute covariance Mermin functions are defined as follows:
\begin{align}
\begin{split}
\mathcal{M}_{0} &=| \texttt{cov}(A_0,B_0)+\texttt{cov}(A_1,B_1) |  \\ 
\mathcal{M}_{1} &=| \texttt{cov}(A_0,B_0)-\texttt{cov}(A_1,B_1) |  \\
\mathcal{M}_{2} &=|\texttt{cov}(A_0,B_1)+\texttt{cov}(A_1,B_0) |  \\
\mathcal{M}_{3} &=|\texttt{cov}(A_0,B_1)-\texttt{cov}(A_1,B_0) |.
\end{split}
\end{align}
 Next, we  construct the following triad of quantities using the above covariance Mermin functions:
\begin{align}
\begin{split}
\Gamma_1&:=\tau\left(\mathcal{M}_0,\mathcal{M}_1,\mathcal{M}_2,\mathcal{M}_3\right)\\
\Gamma_2&:=\tau\left(\mathcal{M}_0,\mathcal{M}_2,\mathcal{M}_1,\mathcal{M}_3\right) \\
\Gamma_3&:=\tau\left(\mathcal{M}_0,\mathcal{M}_3,\mathcal{M}_1,\mathcal{M}_2\right),\label{gi}
\end{split}
\end{align}
where    
\begin{equation}
\tau\left(\mathcal{M}_0,   \mathcal{M}_1,    \mathcal{M}_2,
\mathcal{M}_3\right)  \equiv \Big||\mathcal{M}_0  - \mathcal{M}_1  | -
|\mathcal{M}_2  - \mathcal{M}_3|\Big|,
\end{equation}  
and so on.  
Finally,  \textit{Mermin strength} $\Gamma$ is defined as follows:
\begin{equation}
\Gamma := \min_i \Gamma_i, 
\label{eq:G}
\end{equation}
which takes values as $\Gamma \ge 0$.
This quantity is considered here for the reason that, in addition to witnessing superunsteerability as shown in App.~\ref{certss}, it also provides a necessary and sufficient certification of the presence of \textit{Schr\"{o}dinger strength} in any given correlation, as indicated by a nonzero $SS_{2}$ \footnote{In~\cite{Jeb14}, Mermin discord, which is the same quantity as $\Gamma$, but without invoking the covariance on the Mermin functions, was used to capture the maximal fraction of the Mermin correlation (a maximally steerable correlation~\cite{DDJ+18}), in unsteerable correlations. In the present work, we adopted the quantity of Mermin strength $\Gamma$ as in Eq.~(\ref{eq:G}) to have it vanish for all CQ and QC states, while the quantity of Mermin discord in~\cite{Jeb14} does not always vanish for these states.}. This follows because $\Gamma > 0$ for any given correlation if and only if the correlation admits a decomposition with a nonzero weight of a maximally steerable correlation~\cite{Jeb14,DDJ+18}, implying a nonzero $SS_{2}$~\cite{JDK+19}.
For example, for the Bell diagonal states $\tau_{AB}$,  $\Gamma(\tau_{AB})=2|c_2|$ for the measurements that give rise to $SS_{2}(\tau_{AB})=|c_2|$ in Eq.~(\ref{SSBD}) in Appendix~A. On the other hand, superunsteerability of the given discordant state does not necessarily imply 
a nonvanishing Schr\"{o}dinger strength, see Observation $8$ in Appendix~\ref{certss} for an illustration of this.

We now proceed to define the \textit{optimal decomposition} that gives the \textit{maximal weight of the entangled part} in Eq.~(\ref{LSe}) as follows:
\be \label{LSd}
\rho_{AB}=p \ketbra{\psi_e}{\psi_e}+(1-p) \sigma^{\Gamma=0}_{\texttt{sep}},
\ee
where the separable state, $\sigma^{\Gamma=0}_{\texttt{sep}}$, has Mermin strength,  $\Gamma$ in Eq.~(\ref{eq:G}), vanishing for all measurements.
A separable state or an unsteerable entangled state that has such a decomposition with a \textit{maximal weight} of an entangled state can be used to imply a nonvanishing \textit{Schr\"{o}dinger strength} $SS_{2}$. This follows from the observation that the entangled part in the decomposition can produce an extremal steerable correlation~\cite{DDJ+18}, and this part has the maximal weight. This, in turn, implies that for superunsteerable states that admit the decomposition as in Eq.~(\ref{LSd}),   the weight of the entangled state in Eq.~(\ref{LSd}) is maximized over all possible decompositions to quantify global coherence.

 We can now state the following lemma.
\begin{lem}\label{SSLSd}
    If any superunsteerable state has a nonzero  $\Gamma$ in Eq.~(\ref{eq:G}), it admits a decomposition as in Eq. (\ref{LSd}).
\end{lem}
\begin{proof}
A nonzero  $\Gamma$ in Eq.~ (\ref{eq:G}) implies a nonvanishing \textit{Schr\"{o}dinger strength} $SS_{2}$
which requires a nonvanishing weight of the entangled part in Eq.~(\ref{LSd}).
\end{proof}
Before stating the main result, we present a lemma on \textit{genuine} RSP using mixed states. Note that the payoff function ${\cal  P}$ in the RSP-fidelity is nonzero if and only if Alice is able to effectively prepare a state $\vec{r}$, which is not a completely mixed state. However, this does not necessarily imply a genuine RSP. By genuine RSP, we refer to a nonzero RSP-fidelity that implies that the remotely prepared state $\rho_{\vec{r}}$, with the Bloch vector $\vec{r}$, can be decomposed as 
\begin{equation}\label{gRSP}
\rho_{\vec{r}}=s\rho_{\vec{s}}+(1-s) \rho, 
\end{equation}
with $0< s \le 1$, where $\rho_{\vec{s}}$ is the target state to be prepared and $\rho$ is another state, for any target state. If Alice uses the maximally entangled state, then $s=1$ in the above decomposition can be achieved for any target state. On the other hand, if Alice uses a mixed resource state, then a nonzero RSP-fidelity may not imply that $s>0$ in Eq.~(\ref{gRSP}) is achieved for any target state.
\begin{lem}\label{GRSP}
For any two-qubit state, a nonzero $p$ in Eq.~(\ref{LSd}) is a necessary and sufficient condition to implement a genuine RSP.
\end{lem}
\begin{proof}
For any superunsteerable state that has a \textit{maximal weight of an entangled state} in Eq.~(\ref{LSd}),   the entangled part of the decomposition acts as a \textit{nonlocal resource} to implement a genuine RSP. This follows from the fact that the presence of such a nonlocal resource in the given two-qubit unsteerable state can be used to remotely prepare any target qubit state with a certain fraction, i.e., the remotely prepared state admits a decomposition as in Eq.~(\ref{gRSP}). On the other hand, any steerable state is entangled; hence, it has a nonzero weight of an entangled state to implement a genuine RSP.
\end{proof}

 We have now obtained our second main result.
\begin{thm}
      Global coherence of two-qubit states $\zeta_{AB}$, as indicated by a nonzero Mermin strength $\Gamma(\zeta_{AB})$,  is a necessary and sufficient quantum resource for implementing a genuine RSP using a two-qubit state. 
\end{thm}
\begin{proof}
From Lemma~\ref{SSLSd} and Lemma~\ref{GRSP}, it follows that if any given two-qubit separable state has a nonzero  $\Gamma$ in Eq.~(\ref{eq:G}), then it supports a nonzero RSP-fidelity that implies a genuine RSP. On the other hand, if the two-qubit state has a vanishing  $\Gamma$ in Eq.~(\ref{eq:G}), then it always implies that no genuine RSP takes place.
\end{proof}
The above result can also be seen from the perspective of Giorgi in~\cite{Gio13}, who presented that quantum discord is neither a necessary nor a sufficient resource.  This is due to the following observation. 
\begin{observation} Giorgi~\cite{Gio13}:
 There are states such as the state in Eq.~(\ref{LC_{A,B}}) with a nonzero discord that supports a nonzero RSP-fidelity. Such states have a local creation of discord. Some states with vanishing RSP-fidelity have a nonzero discord. For such states, $c_2=c_3=0$ in Eq.~(\ref{can02q}).  \label{Gio}
\end{observation}
 Because superunsteerability occurs for those discordant states that cannot be created using local operations of the form $\Phi_A \otimes \Phi_B$, it constitutes a sufficient resource to implement a genuinely RSP.
 On the other hand, there are discordant states that are not useful to implement a genuinely RSP. Such states have vanishing $SS_{2}(\zeta_{AB})$ as well as $\Gamma(\zeta_{AB})$ in Eq.~(\ref{eq:G}). Note that the superunsteerable state in  Eq.~(\ref{ssepQ=0}) does not support a nonzero RSP-fidelity, as Giorgi pointed out in~\cite{Gio13}. Such states have $SS_{2}=0$ and $\Gamma=0$, implying that they are not useful for the RSP protocol. Thus, a nonzero $\Gamma$ of the given discordant state is a necessary resource for the RSP protocol.
 
In Fig.~\ref{Fig:Hirarchy1}, we depict the hierarchy of correlations in bipartite states that follows from the present work.  In this hierarchy, the regions $I$-$VI$ represent the separable correlations while the region $VII$ represents the nonseparable correlations, the region $II$ represents the one-way discordant correlations while the regions $III$-$VII$ represent the two-way discordant correlations, the regions $II$-$III$ represent the discordant correlations that have local coherence only while the regions $IV$-$VII$ represent the discordant correlations that have global coherence and the regions $V$ and $VI$ represent the discordant correlations in separable states that exhibit SDI steering, and, finally, the regions $II$-$V$ represent the two-qubit discordant correlations useless for the RSP and the regions $VI$ and $VII$ represent the two-qubit discordant correlations useful for the RSP.    

 \begin{figure}[t!]
\begin{center}
\includegraphics[width=7.5cm]{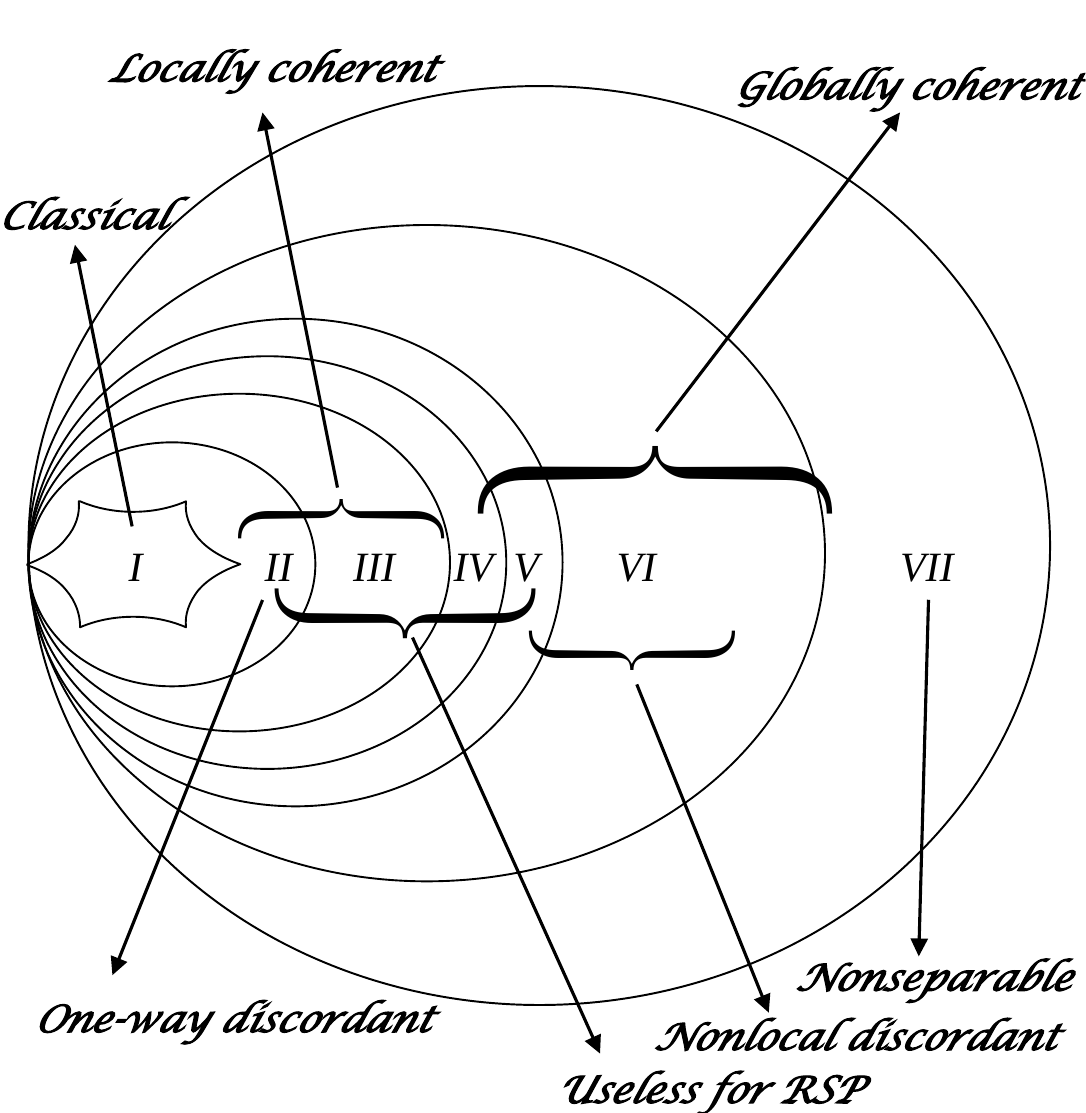}
\end{center}
\caption{Hierarchy of correlations in bipartite quantum states. The region $I$ represents correlations in incoherent states, and the regions $II$ and $III$ represent correlations in locally coherent states that include all one-way discordant and two-way discordant states that can be created locally; the regions $IV$-$VII$ represent correlations in globally coherent states.
 We have found that the regions $V$ and $VI$ (denoted as nonlocal discordant), which are subsets of all two-way discordant correlations that have global coherence, exhibit superunsteerability, implying SDI steerability. Nonlocal discordant correlations in region $V$ with $\Gamma = 0$ are not useful for the RSP.
\label{Fig:Hirarchy1}}
\end{figure}

We now proceed to remark on the quantification of the resource state for the efficiency of the RSP protocol.
 To certify a genuine RSP using a two-qubit separable state,  it suffices to certify that the state is superunsteerable in the two-setting scenario with nonvanishing $\Gamma(\zeta_{AB})$ in Eq.~(\ref{eq:G}). We now make the following observation. From the Schr\"{o}dinger strengths of two-qubit Bell-diagonal states $\tau_{AB}$ as given in Eq.~(\ref{SSBD}), for any Bell-diagonal state, $\tau'_{AB}$, with the given  $SS_{2}(\tau'_{AB})>0$ and $SS_{3}(\tau'_{AB})=0$, the efficiency of the state for the RSP protocol is always lower than that of any other  Bell-diagonal state, $\tau''_{AB}$, with the same nonzero $SS_{2}$ and any nonzero $SS_{3}$. This follows from the fact that the RSP-fidelity of the former state is given by $\mathcal{F}(\tau'_{AB})=\frac{c^2_2}{2}$, while the latter state has $\mathcal{F}(\tau''_{AB})=\frac{c^2_2+c^2_3}{2}$. In~\cite{DLM+12}, in showing the superiority of the superunsteerable state over the entangled state for RSP, quantification of the resource states using quantum discord was used.
 \begin{rem}
  While global coherence, as witnessed by nonvanishing $\Gamma$ in Eq.~(\ref{eq:G}),
  certifies the genuinely quantum resource for the RSP protocol in~\cite{DLM+12}, quantum discord and the Schr\"{o}dinger strengths of the resource states can be used to provide quantification/comparison of the efficiencies of the protocol using different resource states. 
  \end{rem}

\section{Conclusions and Discussions}\label{concDisc} 
In this work, we have been motivated to study \textit{which aspect of bipartite coherence} implies \textit{genuinely quantum correlation} in discord. To this end, we have defined \textit{global coherence} of bipartite states, which is a specific form of bipartite coherence that is not local coherence in any of the subsystems or both subsystems alone. Global coherence of separable states is equivalent to the witness of discord by the rank of the correlation matrix of the density matrix introduced in~\cite{DVB10}.   In case of \textit{global coherence as a resource},   \textit{any local operations} of the form $\Phi_A \otimes \Phi_A$, which may create coherence locally, are \textit{free operations}. To rigorously formulate global coherence, we pointed out a strong link between resource frameworks of global coherence in bipartite and single-partite states, just as the strong link between entanglement and coherence shown in~\cite{SSD+15}.

We have then addressed the question of whether \textit{genuinely quantum correlation} in discord, as indicated by \textit{global coherence}, is required for  \textit{operational tasks}. To this end, we have considered the task of \textit{SDI steering} for which shared randomness is not freely available. We have then shown that any discordant state, which does not have global coherence, cannot be used to demonstrate \textit{superunsteerability}. Thus, we have established that global coherence is necessary for SDI steering tasks. 

The result obtained in Theorem~\ref{thm} also holds for superlocality~\cite{DW15,JAS17} in the context of SDI Bell scenarios.
We now address the improvements achieved in this result of the present investigation over the previous research in~\cite{BPP15,JAS17, JD23,JDS_PRA23}. 
In~\cite{BPP15}, \textit{genuinely quantum correlation} in separable states was already proposed to be in discord with $L_R > d_{\min}$, as they cannot be created via local operations. On the other hand, we have identified that such genuinely quantum correlation is a specific form of bipartite coherence that is not local to any of the subsystems or both of the subsystems and acts as a necessary resource for superunsteerability or superlocality. 
In~\cite{JDS_PRA23}, the question of whether the one-wayness property of quantum discord can be used to demonstrate the asymmetry of superunsteerability was explored. Using Theorem~\ref{thm}, it follows that all one-way discordant states cannot be used to demonstrate superunsteerability. Thus,  the occurrence of superunsteerability asymmetrically remains to be studied.
On the other hand, in~\cite{JAS17, JD23},  two-way discord as a necessary resource for superlocality was studied, and it remained unknown whether every two-way discordant state can be used to demonstrate superlocality.  From Theorem~\ref{thm}, it follows that a two-way discord is not sufficient for superlocality. In addition to this,  we have also identified the necessary condition on two-way discordant states to demonstrate superlocality.

Finally, to answer the aforementioned question of whether \textit{global coherence} is required for  \textit{operational tasks} in the context of quantum information processing using separable states, we have considered the quantum communication task of \textit{RSP} using two-qubit correlations~\cite{DLM+12,Gio13}. We have then identified the necessary and sufficient quantum resource for this task by the witness of global coherence called ``Mermin strength''.

Taking quantum discord as a resource for quantum information processing based on previous proposals, such as in~\cite{DLM+12}, may be misleading due to the local creation of the resource. Although in~\cite{DLM+12} quantum discord was proposed as a resource for RSP,  the experimental results of~\cite{DLM+12} are impressive and are not misleading from the perspective of the present work: the states chosen in the experiments are a separable state having superunsteerability and an entangled state, and indeed it is very surprising that the separable state outperforms the entangled state, as shown in the experiments of~\cite{DLM+12}.  
 
 In contrast to the experimental demonstration in~\cite{DLM+12},  with a general framework for the RSP in~\cite{HTM14}, it was shown that separability can never outperform entanglement in the quantum communication task. In this framework, a different version of RSP-fidelity based on linear fidelity and general decoding operations on Bob's side was studied. In showing the superiority of entanglement over separability in the RSP, an assumption of one-way local operation and classical communication native to entanglement theory is utilized. This assumption and the RSP-fidelity based on linear fidelity \textit{convexifies} the quantum resource for the RSP. This convexification of the quantum resource implies that all separable states are free resources for the RSP protocol. Even with this demonstration of~\cite{HTM14}, the experimental results of~\cite{DLM+12} are still surprising, as these two RSP protocols using the two-qubit states are in two different resource-theoretic frameworks: the latter is a nonconvex resource theory, whereas the former is a convex resource theory.  For the practical implementation of RSP, the RSP protocol in~\cite{DLM+12} seems more relevant than the protocol in~\cite{HTM14}, which shows the superiority of entanglement over separability.

\section*{Acknowledgements}
This work was supported by the National Science and Technology Council (NSTC), the Ministry of Education (MOE) through the Higher Education Sprout Project NTU-113L104022-1, and the National Center for Theoretical Sciences (NCTS) of Taiwan. H.- Y.~K. is supported by NSTC, under Grant
No.~NSTC 112-2112-M-003- 020-MY3, and MOE through the Higher Education Sprout Project of National Taiwan Normal University (NTNU). 
H.-C.~Cheng is supported by the Young Scholar Fellowship (Einstein Program) of NSTC under Grants No.~NSTC 112-2636-E-002-009, No.~NSTC 113-2119-M-007-006, No.~NSTC 113-2119-M-001-006, No.~NSTC 113-2124-M-002-003, and No.~NSTC 113-2628-E-002-029 and by the Yushan Young Scholar Program of MOE under Grants No.~NTU-113V1904-5. 
H.-S.G. acknowledges support from NSTC under Grants No. NSTC 113-2112-M-002-022-MY3, No. NSTC 113-2119-M-002 -021, and No. NSTC 114-2119-M-002-017-MY3, from the US Air Force Office of Scientific Research under Award Number FA2386-23-1-4052 and the support of Taiwan Semiconductor Research Institute (TSRI) through the Joint Developed Project (JDP), and the support from the Physics Division, NCTS. 
H.-S.G. and H.-C.~Cheng acknowledge support from the research project ``Pioneering Research in Forefront Quantum Computing, Learning and Engineering'' of National Taiwan University (NTU) under Grant No. NTU-CC- 113L891604  No.~NTU-CC-114L8950, and No.~NTU-CC-114L895004, as well as Grants No.~NTU-CC-113L891605 
and No.~NTU-CC-114L895005, respectively.
H.-S.G. and H.-C.~Cheng also acknowledge the support from the “Center for Advanced Computing and Imaging in Biomedicine (NTU-113L900702 and NTU-114L900702)” through the Featured Areas Research Center Program within the framework of the MOE Higher Education Sprout Project.

\appendix

\section{Schr\"{o}dinger strength}\label{AppSS}
Similarly to the quantification of standard steerability based on the steering cost~\cite{DDJ+18}, a quantifier termed \textit{Schr\"{o}dinger strength} was introduced in~\cite{JDK+19}, which can be used to quantify SDI steering~\cite{JDS24}. To introduce these two quantifiers,
consider $n$-setting steering scenario to observe a box $P(a,b|x,y)$, where $x$ $\in$ $\{0, 1, 2, ..., n-1\}$ and $y$ $\in$ $\{0, 1, 2, ..., n-1\}$. In this scenario, any correlation $P(a,b|x,y)$ can be decomposed into a convex mixture of a steerable part and an unsteerable part,
\be \label{steerstrength}
 P(a,b|x,y)=p P_{S}(a,b|x,y)+(1-p) P_{US}(a,b|x,y),    
\ee
where $P_{S}(a,b|x,y)$ is  an extremal steerable correlation and $P_{US}(a,b|x,y)$  is an unsteerable correlation which may be superunsteerable; $0 \leq p \leq 1$.  In the given scenario, an extremal steerable correlation cannot be expressed as a convex mixture of other correlations. Thus, it maximally violates the steering inequality. 
The weight of the extremal steerable correlation, $P_{S}(a,b|x,y)$, minimized over all possible decompositions of the form (\ref{steerstrength}) is called the steering cost of the correlation $P(a,b|x,y)$ providing a measure of quantum steering \cite{DDJ+18}.

On the other hand, the weight of the extremal steerable correlation, $P_{S}(a,b|x,y)$, in Eq.~(\ref{steerstrength}) maximized over all possible decompositions is called the \textit{Schr\"{o}dinger strength}  (SS) of the correlation $P(a,b|x,y)$:
\be \label{SchS}
SS_{n} \Big( P(a,b|x,y) \Big):=\max_{\text{decompositions}}p.
\ee 
Schr\"{o}dinger strength as defined above has the following properties:
\begin{itemize}
\item $0 \leq SS_{n} \leq 1$ (since $p$ is a valid probability, $0 \leq p \leq 1$).
\item $SS_{n}=0$ for all correlations that have a decomposition as in Eq. (\ref{LHV-LHS}) with $d_\lambda \le d_A$.
\item $SS_{n}>0$ implies that the correlation does not have a decomposition as in Eq.~ (\ref{LHV-LHS}) with $d_\lambda \le d_A$.  
\item $SS_{n}$ has nonconvexity, i.e., there exist correlations  
$P_1(a,b|x,y)$ and $P_2(a,b|x,y)$ such that 
\begin{multline}
SS_{n}\left(p P_1(a,b|x,y) + (1-p) P_2(a,b|x,y)\right)  \\
\ge p SS_{n}\left( P_1(a,b|x,y) ) + (1-p) SS_{n}(P_2(a,b|x,y)\right),
\end{multline}
for $0 \le p \le 1$.
\end{itemize}

In~\cite{JDK+19}, the Schr\"{o}dinger strengths of two-qubit Bell-diagonal states have been studied in two- and three-setting scenarios. 
The two-qubit Bell diagonal states $\tau_{AB}$ are the states in Eq.~(\ref{can02q}) with local Bloch vectors $\vec{a}$ and $\vec{b}$ being zero vectors, i.e., all two-qubit states with maximally mixed marginals.
 The Schr\"{o}dinger strengths of two-qubit Bell-diagonal states with $c^2_1\ge c^2_2 \ge c^2_3$  in the two- and three-setting scenarios are, respectively, evaluated as
 \begin{align} \label{SSBD}
 SS_{2}(\tau_{AB})&=|c_2| \quad \& \quad  SS_{3}(\tau_{AB})=|c_3|.
 \end{align}
 In~\cite{JDS24}, for the quantification of the SDI steerability of Bell-diagonal states by the Schr\"{o}dinger strengths,  a geometric characterization has been given by using the quantum steering ellipsoid~\cite{JPJR14}. See Fig.~\ref{Fig:SSWer} for the variation of $SS_{2}(\rho_W)=p$ in the space of all two-qubit states. The fraction of the entangled state $p$ in Eq. (\ref{Wers}) also provides the quantification of global coherence.

 \begin{figure}[t!]
\begin{center}
\includegraphics[width=7cm]{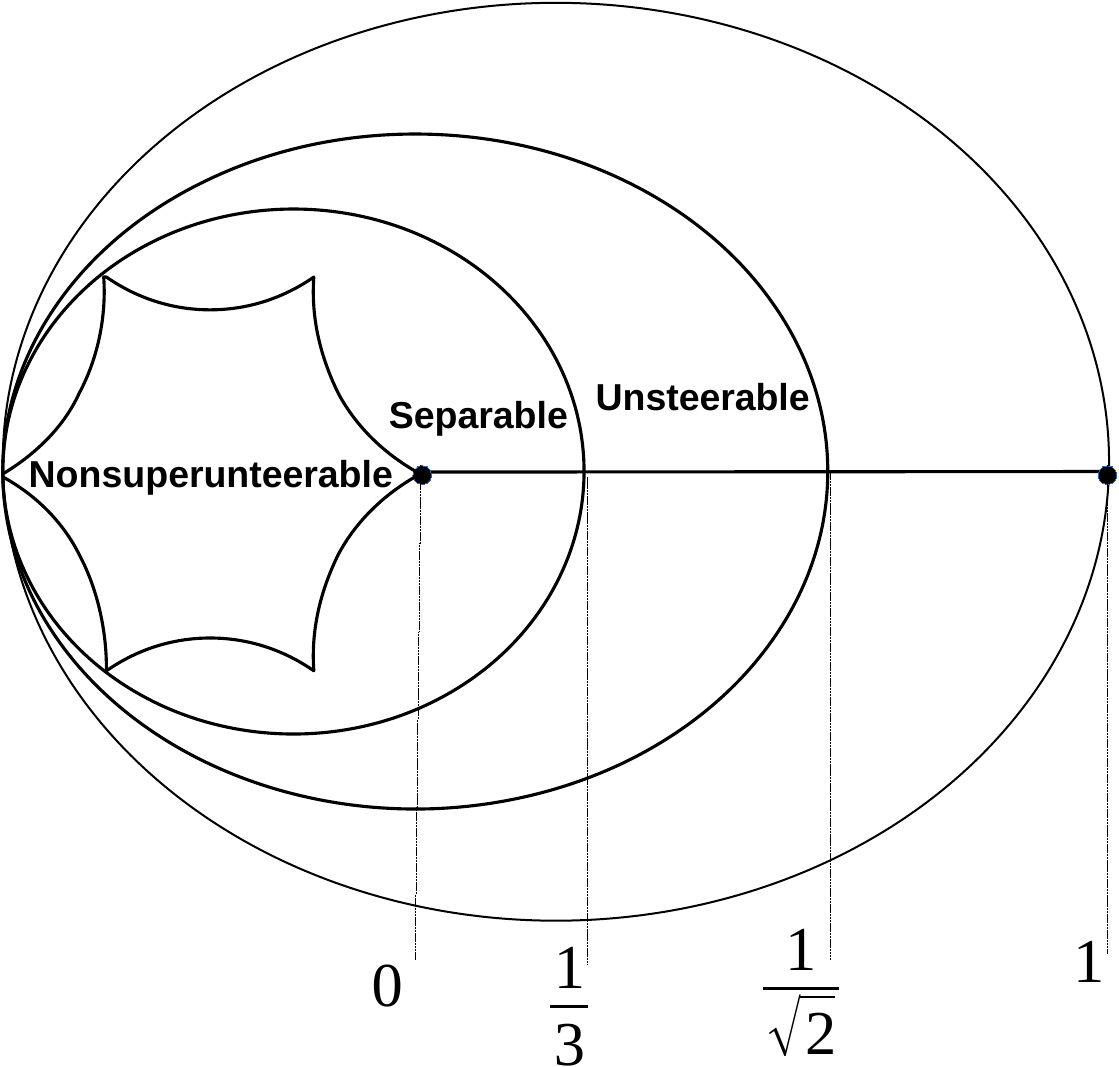}
\end{center}
\caption{An abstract representation of the convex set of correlations in separable and unsteerable states in the two-setting scenario for two-qubit systems and the nonconvex set of correlations in the nonsuperunsteerable states. The straight line represents the Werner state given by Eq. (\ref{Wers}), which has the Schr\"{o}dinger strength,  $SS_{2}(\rho_W)=p>0$ for any $p>0$ \cite{JDK+19}. In the two-setting scenario, the Werner state has SDI steerability for any $p>0$, on the other hand, it has standard steerability for $p>1/\sqrt{2}$.   
\label{Fig:SSWer}}
\end{figure}


\section{Mermin strength as a witness of Schr\"{o}dinger strength}\label{certss}
We now demonstrate the following lemma for the vanishing $\Gamma$ defined in Eq.~(\ref{eq:G}).
\begin{lem}
For any two-qubit  CQ state as in Eq. (\ref{CQ}) or two-qubit QC state, $\Gamma=0$ for all measurements.
\label{propo01}
\end{lem}
\begin{proof}
Any two-qubit CQ state given by Eq. (\ref{CQ}) can be decomposed as follows:
\begin{multline}
    \rho_{CQ}=\frac{1}{4}\big[ \openone_2 \otimes \openone_2 + (p_0-p_1) \hat{r} \cdot \vec{\sigma} \otimes  \openone_2  \\
+  \openone_2 \otimes \left(p_0 \vec{s}_0+  p_1 \vec{s}_1 \right)\cdot \vec{\sigma}  \\
+\hat{r} \cdot \sigma \otimes \left( p_0\vec{s}_0-p_1 \vec{s}_1 \right)\cdot \vec{\sigma} \big],
\end{multline}
where $\hat{r}$ is the Bloch vector of the projectors $\ketbra{i}{i}^A$ and  $\vec{s}_i$ are the Bloch vectors of the quantum states $\rho^B_i$. Since we only make a qubit assumption on Alice's side, her measurements are \textit{a priori}  POVM with elements given by
\be
M^A_{a|x}=\gamma^{A}_{a|x} \openone_2+ (-1)^{a} \frac{\eta^A_{x}}{2} \hat{u}^A_{x} \cdot \vec{\sigma}, \label{POVMalice}
\ee
where  $\gamma^{A}_{0|x}+\gamma^{B}_{1|x}=1$ $\forall$ $x$ and $0 \leq \gamma^{A}_{a|x} \pm \frac{\eta^A_{x}}{2} \leq 1$ $\forall$ $x, a$.
On the other hand, Bob's measurements are also  POVMs with elements given by
\be
M^B_{b|y}=\gamma^{B}_{b|y} \openone_2+ (-1)^{b} \frac{\eta^B_{y}}{2} \hat{v}^B_{y} \cdot \vec{\sigma}, \label{POVMbob}
\ee
where  $\gamma^{A}_{0|y}+\gamma^{B}_{1|y}=1$ $\forall$ $y$ and $0 \leq \gamma^{B}_{b|y} \pm \frac{\eta^B_{y}}{2} \leq 1$ $\forall$ $y, b$. Now, it can be easily checked that the above measurement settings always lead to $\Gamma=0$ for any CQ state $\rho_{CQ}$. 
Similarly, it can be seen that any QC state always gives $\Gamma=0$.
\end{proof}

From the above lemma, we now obtain the following result.
\begin{prop}
Suppose $\Gamma$ in Eq.~(\ref{eq:G}) takes a nonzero value for any given unsteerable correlation. Then the correlation has superunsteerability. 
\end{prop}
\begin{proof}
Note that in the context of the two-setting steering scenario with dichotomic measurements,  any unsteerable correlation can be reproduced by a CQ state of the form $\sum^{d_\lambda-1}_{\lambda=0} p_\lambda \ketbra{\lambda}{\lambda} \otimes \rho_\lambda$ \cite{MGH+16},
 where $\{\ketbra{\lambda}{\lambda}\}$ forms an orthonormal basis in  $\mathbb{C}^{d_\lambda}$, with $d_\lambda \le 4$. In other words, 
 the unsteerable correlation admits the decomposition as in Eq. (\ref{LHV-LHS}) with $d_\lambda \le 4$.
 Here, the dimension of the hidden variable is upper bounded by $4$ since any unsteerable correlation corresponding to this scenario can be simulated by shared classical randomness of dimension $d_\lambda \le 4$ \cite{DW15,DBD+18}.
It then follows that any unsteerable correlation produced from a two-qubit state that requires a hidden variable of dimension $d_\lambda \le 2$
to provide an LHV-LHS model (i.e., any unsteerable correlation produced from a two-qubit state that is not superunsteerable)
can be simulated by a two-qubit state that admits the form of the CQ state given by Eq. (\ref{CQ}).
Thus, for any such unsteerable box,
$\Gamma=0$. On the other hand,  if any unsteerable correlation produced from a two-qubit state has $\Gamma>0$,  then the correlation does not arise from a CQ or QC state. Hence, this correlation requires a hidden variable of dimension $d_\lambda>2$ to provide an LHV-LHS model (i.e., the correlation is superunsteerable).
\end{proof}

To illustrate that $\Gamma$ is needed to capture the presence of a nonzero $SS_{2}$ in any given superunsteerable correlation, let us make the following observation.  
\begin{observation}
    There are superunsteerable states that cannot be used to imply a nonvanishing \textit{Schr\"{o}dinger strength}.
\end{observation}
To illustrate the above, consider the discordant state as given by Eq.~($13$) in~\cite{Gio13}. It has global coherence since  $L_R$ in Eq.~(\ref{LR}) satisfies $L_R=3$. This state can be written in the form given by Eq.~ (\ref{can02q}) as follows:
\begin{multline} \label{ssepQ=0}
\rho_{AB}=\frac{1}{4} \Big( \openone_2 \otimes \openone_2 + 0.4 \sigma_1 \otimes \openone_2 +0.4 (\openone_2 \otimes \sigma_1 - \openone_2 \otimes \sigma_3)    \\ +0.2 \sigma_3 \otimes \sigma_3  \Big).
\end{multline}
To show that the above state has superunsteerability, we consider the certification of superunsteerability of two-qubit states, with two dichotomic measurements on both subsystems, as introduced in~\cite{JD23}, using a nonlinear witness that certifies quantum discord both ways. This witness is given by a determinant which is written in terms of the conditional probabilities on Bob's side $\{p(b|a;x,y)\}$ as follows:
\begin{widetext}
\begin{align}
	Q&=\left|\begin{array}{cc}p(\mathrm{0|0;0,0})-p(\mathrm{0|1;0,0}) & p(\mathrm{0|0;1,0})-p(\mathrm{0|1;1,0})\\ 
		p(\mathrm{0|0;0,1})-p(\mathrm{0|1;0,1}) & p(\mathrm{0|0;1,1})-p(\mathrm{0|1;1,1})\end{array}\right|. \label{QC}
\end{align}
\end{widetext}
A nonzero value of $Q$ by a bipartite correlation produced by a two-qubit state implies that it cannot be simulated by a two-qubit CQ state or a QC state. From this, it follows that the quantity $Q$ can be used to witness superunsteerability or steerability. 
It has been checked that the state (\ref{ssepQ=0}) has a nonvanishing $Q$ in Eq.~(\ref{QC}),  for instance,  for the measurements with POVM elements given by
$M^A_{0|0}=M^B_{0|0}=\ketbra{0}{0}$ and $M^A_{0|1}=M^B_{0|1}=\ketbra{+}{+}$ it gives rise to $Q \approx 0.0381$. This implies that the correlation that gives rise to such a nonzero value does not admit an LHV-LHS  with $d_\lambda=2$. Therefore, the state (\ref{ssepQ=0}) has superunsteerability. On the other hand, it has been checked that this superunsteerable state has vanishing $\Gamma$ in Eq.~(\ref{eq:G}) for all measurements. This implies that the superunsteerable state cannot be used to imply a nonzero Schr\"{o}dinger strength $SS_{2}$. This may be compared with an example of a Bell-nonlocal bipartite state presented in~\cite{VPO+20}, which does not admit a nontrivial singlet fidelity for device-independent certification of the singlet state~\cite{Jed16}. 


\end{document}